\documentclass[12pt]{iopart}

\usepackage{cite}
\usepackage{graphicx}
\usepackage{graphics}
\usepackage{iopams} 
\usepackage{color} 

\newtheorem{theorem}{Theorem} 
\newtheorem{lemma}{Lemma}
\newenvironment{proof}[1][Proof]{\begin{trivlist}
\item[\hskip \labelsep {\bfseries #1}]}{\end{trivlist}}

\newcommand{\qed}{\nobreak \ifvmode \relax \else
      \ifdim\lastskip<1.5em \hskip-\lastskip
      \hskip1.5em plus0em minus0.5em \fi \nobreak
      \vrule height0.75em width0.5em depth0.25em\fi}

\newcommand{\sech}{\rm sech\,}
\newcommand{\csch}{\rm csch\,}
\newcommand{\arctanh}{\rm arctanh\,}

\def\be{\begin{equation}}
\def\eea{\end{eqnarray}}
\def\ee{\end{equation}}
\def\bea{\begin{eqnarray}}
\def\ea{\end{array}}
\def\ba{\begin{array}}

\newcommand{\exval}[1]{\mbox{$\left\langle \, {#1}\, \right\rangle$}}

\newcommand{\bel}[1]{\begin{equation}\label{#1}}

\newcommand{\deriv}[1]{\mbox{$\displaystyle\frac{d}{d{#1}}$}}
\newcommand{\pderiv}[1]{\mbox{$\displaystyle\frac{\partial}{\partial{#1}}$}}

\newcommand{\NNS}[1]{\mbox{$^{ \left\langle \, {#1}\, \right\rangle}$}}

\def\zzz{{\mathchoice {\hbox{$\sf\textstyle Z\kern-0.4em Z$}}
{\hbox{$\sf\scriptstyle Z\kern-0.3em Z$}}
{\hbox{$\sf\scriptscriptstyle Z\kern-0.2em Z$}}
{\hbox{$\sf\textstyle Z\kern-0.4em Z$}}}}

\begin{document}

\title{The hypergeometric series for the partition function of the 2-D
  Ising model}

\author{G. M. Viswanathan}

\address{Department of Physics
and {National Institute of Science and Technology of Complex Systems,}
Universidade Federal do Rio Grande do Norte, 59078-970
Natal--RN, Brazil}

\begin{abstract}
In 1944 Onsager published the formula for the partition function of
the Ising model for the infinite square lattice.  He was able to
express the internal energy in terms of a special function, but he
left the free energy as a definite integral.
Seven decades later, the partition function and free energy have yet to be
written in closed form, even with the aid of special functions.
Here we 
evaluate the definite integral
explicitly, using hypergeometric series.
Let $\beta$ denote the reciprocal temperature, $J$ the coupling and
$f$ the free energy per spin.
We prove that
\mbox{$ - \beta f = \ln(2 \cosh 2K) - \kappa^2
 ~\! {_4F_3} \big[~\!
    ^{1,~1,~3/2,~3/2} _{~~~2,~2,~2}
;16 \kappa^2
    ~\!\big] 
~\!,
$}
where $_p F_q$ is the generalized hypergeometric function,
$K=\beta J$, and $2\kappa= \tanh 2K \sech 2K$.

\end{abstract}


\bigskip
\bigskip
\bigskip
\bigskip

\bigskip
\bigskip
\bigskip
\bigskip

\bigskip
\bigskip
\bigskip
\bigskip

\bigskip
\bigskip
\bigskip
\bigskip

\noindent Journal version: \hfill
http://dx.doi.org/10.1088/1742-5468/2015/07/P07004

\submitto{Journal of Statistical Mechanics (JSTAT)  }

\maketitle

The Ising model was originally proposed by Lenz to describe magnetism
\cite{ising}.  His student Ising solved the eponymous model in one
dimension (\mbox{1-D}) in his doctoral thesis.  Since then,
it has become one of the most important models in the history of
statistical physics
\cite{huang,barrybook,baxter,salinas,feynman,stanley}.  The 2-D case
for zero external magnetic field was not fully solved until 1942, and
the results were published only in 1944 \cite{onsager}.  The 3-D case
remains unsolved.  Onsager's derivation of the partition function of
the 2-D Ising model is often described as a mathematical tour de
force. He was able to express the internal energy in terms of a
special function, but he left the partition function and the free
energy in terms of a definite integral \cite{onsager}.  Ever since, a
common assumption has been that there is no succinct way to write the
partition function in closed form, even using special functions.
The required 
integral is indeed not easy to evaluate
explicitly, but  we show that it is possible to do so.
In fact, Onsager had already expanded the free energy in a power
series~\cite{onsager}, however he did not 
recognize it as a known special
function (at least not until late in his life, see below).
Here we continue the calculation where he left off and show that his
power series is hypergeometric.  

{

Readers unfamiliar with special functions may question what is gained by
trading a definite integral for a special function.  One is just as
complicated as the other, it could be argued.  Consider the following
helpful 
analogy with trigonometric functions. For $|x| \leq 1$ the integral
\[ 
\int_0^x {dy \over \sqrt{1-y^2} } \]
can be written as $\arcsin x$.  Neither expression contains more
information than the other.  Yet most readers will agree that $\arcsin
x$ is preferrable to the integral, mainly because trigonometric
functions are well understood.  The same is true of special functions.
In fact, the Hungarian mathematician Paul Tur\'an thought that
``special functions'' should instead be renamed {\it useful
  functions\/,} according to Askey (see
refs.~\cite{cipra1998,andrewsbook}).  They are ubiquitous and arise in
a variety of physical problems.  For example, hypergeometric functions
appear very naturally in mathematical physics
\cite{arXiv:1207.2815,glasser-lamb-jphysa,joyce-jphysa,glasserpaper},
statistical mechanics
\cite{doi:10.1088/1751-8113/44/38/385002,doi:10.1088/1751-8113/45/7/075205},
and number theory \cite{ramanujan-j}.
}

We briefly review  the square lattice Ising model with
symmetric coupling.  Consider a two dimensional square lattice where
at each point $i$ of the lattice is located a spin-{\small 1/2}
particle.  Each spin $\sigma_i$ can assume only 2 values:
$\sigma_i=\pm 1$. 
The Ising model Hamiltonian as a function of a spin configuration
\mbox{$\sigma=(\sigma_1,\sigma_2,\ldots )$} is given by

\be H(\sigma)= -J \sum_{\NNS{i,j}} \sigma_i \sigma_j ~.  
\ee 
Here $\exval{i,j}$ represents the set of lattice
points $i,j$ which are nearest neighbors.
The sum
should avoid double counting, so that
pairs $\exval{i,j}$ and $\exval{i,j}$ are not counted
separately.  The constant $J$ is known as the coupling and is an
interaction energy.

We are interested in the thermodynamic limit, but let us initially
consider the partition function for a finite $L \times L$ system with
$N=L^2$ spins.
Let  $T$ be the thermodynamic temperature, $ k_{\mbox{\tiny B}} $
Boltzmann's constant and let  $\beta=1/(k_{\mbox{\tiny B}}T)$. The
partition function $Z_N(\beta)$ is then defined as the sum over all
possible spin configurations of $\exp\left(-\beta H(\sigma)\right)$:
\be
Z_N(\beta) = \sum_{\sigma} e^{-\beta H} 
~.
\ee
So $Z_N$ can also be thought of as the two-sided Laplace transform of
the degeneracy $\Omega(E)$ of the energy level $E$.  The free energy
is given by
$
F= - k_{\mbox{\tiny B}} T \ln Z_N
$
and  the free energy
per spin in the thermodynamic limit is given by
$
f=  - k_{\mbox{\tiny B}} T \ln  \lambda ~,  
$
where, 
\be \lambda= \lim_{N\to \infty } Z_N^{1/N} ~.   \ee

Onsager referred to $\lambda$ as the ``partition function per atom''
and henceforth we will refer to $\lambda$ simply as the partition
function.  In the seminal work of 1944, he derived the exact partition
function,
\begin{eqnarray}
  \ln \lambda &= \ln{2 \cosh(2K)}
\nonumber
\\
& \quad
+ {1\over 2\pi^2}
\int_0^\pi
\int_0^\pi
\ln (1-4\kappa \cos \omega_1 \cos \omega_2)
~d\omega_1 d\omega_2
 ~, \label{eq-f-ons} 
\end{eqnarray}
where 
\bel{eq-k-def}
2\kappa = \tanh(2K) \sech(2K)
 ~.
\ee  
One of the two integrals can be
evaluated, yielding

\be
\ln \! \left({\lambda \over 2 \cosh 2K}\right)=
{1\over 2 \pi}
\int_0^\pi
\!
\ln \! \left({1+\sqrt{(1- (4\kappa\sin \varphi)^2} \over 2} \right)
\! d\varphi
\label{eq-worlfframaeoif} 
\ee
and versions of expressions (\ref{eq-f-ons}) and
(\ref{eq-worlfframaeoif}) are those found in the textbooks.  For
technical details and historical
context, see
refs. \cite{feynman,barrybook,huang,baxter,onsager,sherman,burg,bru,montroll,cipra,costa,peierls,kw,waerden,kac,hurst-spin,temperley,kasteleyn,vdovi,romance,baxter-enting,stanley,salinas,ising,fisher-ons,lieb}.

Onsager wrote that the integral appearing in the partition function
and free energy 
``cannot be expressed in closed form.''
It is true that the required closed form cannot be found among the
commonly tabulated integrals, even in terms of special functions.
Nevertheless, we will show below that it is certainly possible to
express the integral in terms of a special function. Whether or not it
is of ``closed form'' is a matter of context and
convention. Traditionally, the term ``closed form'' does not include
special functions. However, Onsager himself considered some special
functions to be in closed form, for example the elliptic integrals
(see \cite{onsager} and the discussion below).  Indeed, in the context
of the Ising model and statistical mechanics, many special functions
form part of the repertoire of ``closed form'' expressions.  From this
point of view, our result represents an advance.  Before we state our
claim, we review a few more relevant facts.

Although the free energy has never
before been expressed in terms of special functions, yet it is 
possible~\cite{onsager} to express the internal
energy per spin, defined by
\bel{eq-ons-u-def}
u=- \pderiv{\beta} \ln \lambda 
~,
\ee
in
terms of an elliptic integral:
\be
\label{eq-ons-u}
u= -J (\coth 2K) \left(1 + {2\over \pi}\left(2(\tanh 2K)^2-1\right) {\mathsf K}
(4\kappa) \right) ~. 
\ee
Here we have chosen the sans serif letter $\mathsf K$ to distinguish
the complete elliptic integral of the first kind ${\mathsf K}(k)$ from
the reduced reciprocal temperature $K=\beta J$.  We use the same
notation for the elliptic integral adopted in
\cite{onsager,analysisbook,andrewsbook}, where the argument $k$ of $\mathsf K(k)$
is the {\it elliptic modulus} and not the {\it parameter} (the latter
defined as $m=k^2$).  There are several conventions in use, so this is
an important point to note.  The definition is thus
\bel{eq-K-int}
{\mathsf  K}(k) = 
\int_0^{\pi/2} {d\theta \over \sqrt{1- (k\sin \theta)^2~}} 
~. 
\ee

Elliptic integrals can be expressed in terms of $_pF_q$ generalized
hypergeometric functions \cite{andrewsbook}.  
For example, the complete elliptic integral of the first kind 
above is equivalently given  by 
\bel{eq-ellipticK-2f1}
{\mathsf  K}(k) = \left({\pi \over 2}\right) 
{_2F_1}\left[\ba{c} \frac{1}{2} \frac{1}{2} \\ 1 \ea;k^2\right]
~. 
\ee
We briefly explain this notation. 
A $_pF_q$ function has a power
series $\sum c_n x^n $ such that the ratio $c_{n+1}/c_n$ of
successive coefficients is a rational function of $n$, i.e. a
ratio of polynomials in the degree $n$ of the summed monomials.
 The numbers $p$ and  $q+1$ give the
degrees of the polynomials of the numerator and denominator,
respectively.
 Let the
Pochhammer symbol $(x)_n$ denote the
rising factorial,
\begin{eqnarray}
(x)_0&=1 \nonumber \\
(x)_n&= (x)_{n-1} (x+n-1)~; \quad \quad n=1,2,3\ldots 
\end{eqnarray}
Equivalently, in terms of the gamma function $\Gamma(x)$, the
Pochhammer symbol is given by $(x)_n=\Gamma(x+n)/\Gamma(x)$.
The $_pF_q$ function is then concisely defined as follows:
    \bel{eq-pfq-def}
_pF_q\left[ \ba{c}{a_1,a_2,\dots, a_p} \\ {b_1,b_2,\dots,b_q}\ea ;x \right]
= \sum_{n=0}^\infty
{(a_1)_n(a_2)_n \dots (a_p)_n  \over (b_1)_n(b_2)_n \dots (b_q)_n  }
 ~{x^n \over
         {n!}}
~.
\ee 
The condition $p=q+1$ separates the two distinct regimes $p<q+1$ for
which the $_pF_q$ function is entire, and $p>q+1$ when the radius of
convergence is zero. When $p=q+1$ exactly, the radius of convergence
is 1 (with convergence on the unit circle a somewhat delicate
issue). In our case, we will find $p=4$ and $q=3$ so that $p=q+1$.
Similarly, for the elliptic integral in (\ref{eq-ons-u}) we can see
from (\ref{eq-ellipticK-2f1}) that $p=2$ and $q=1$.  So these
functions have unit radius of convergence.  The singularity of
$\mathsf K(4\kappa)$ at $4\kappa=1$ in Eq.~(\ref{eq-ons-u}) is precisely what is
responsible for critical point of the phase transition in the 2-D
Ising model.  The radius of convergence corresponds to the celebrated
critical temperature of the phase transition, first found by Kramers
and Wannier \cite{kw} in 1941.

Having reviewed the necessary definitions, we state our main result:

\begin{theorem}
Onsager's partition function $\lambda$ in (\ref{eq-f-ons}) can be
written in terms of a hypergeometric function as $\lambda_*$, where

\begin{eqnarray}
\ln \lambda_*
&= \ln(2 \cosh 2K) 
- \kappa^2 ~{_4}F_3
\left[ \ba{c}{1,1,{3\over 2},{3\over 2}} \\ {2,2,2}\ea 
; 16 \kappa^2 \right]  \label{eq-claim2}~,
\end{eqnarray}
and where  $\kappa$ is defined by (\ref{eq-k-def}).

\label{th-main}
\end{theorem}

{
 The $_4F_3$ function above cannot be reduced to a sum of functions
 of simpler type. For example, comparing series one can show that the
 following $_3F_2$ function is reducible in terms of simpler
 functions:
 \[\
 _3F_2\left[1,\frac 3 2,\frac 3 2; 2, 2;x\right] = -\frac 4 x + \frac 8 {\pi x} {\mathsf K}(x^{1/2}) ~.  
 \] 
 In contrast, the $_4F_3$ in Theorem 1 cannot be reduced in this
 manner, i.e. it cannot be expressed as a sum of functions of simpler
 type (see also ref. \cite{newreferee-ref}).  }

{ The $_4F_3$ hypergeometric function in Theorem \ref{th-main} is not
  entirely unexpected in fact, because the integral in
  (\ref{eq-f-ons}) is of a type known as a Mahler measure
  \cite{arXiv:1207.2815,glasserpaper}, that can lead to $_4F_3$
  functions. 
The Mahler measure associated with the integral in Onsager’s formula
can be obtained from Eq. (39) in ref.~\cite{referee2paper-rogers},
Eq. (17) in ref.~\cite {arXiv:1207.2815} or Eq. (12) in
ref.~\cite{glasserpaper}.
}

  We have very recently learnt \cite{pc} that, in the 1970s, Glasser
  and Onsager working together arrived at a similar (but different)
  expression to the one above in Theorem~\ref{th-main}.  They used a
  $_4F_3$ function as well as the complete elliptic integral of the
  second kind, $\mathsf E(k)$.  However, they did not publish their
  result.
Considering the potentially broad interest in this fascinating
piece of 
historical information, below we restate their previously unpublished
result.

\begin{theorem}[Glasser and Onsager]
The partition function $\lambda$ in (\ref{eq-f-ons}) can be
rewritten as $\lambda_\star$, where
\label{th-glasser}
\begin{eqnarray}
&\ln \lambda_\star = \ln(2 \cosh 2K) 
\nonumber
\\
& 
\quad\quad
\quad\quad
- {1\over 2} 
+ {1\over\pi} {\mathsf E(4 \kappa)} 
+ \kappa^2 ~{_4}F_3
\left[ \ba{c}{{1\over 2},1,1,{3\over 2}} \\ {2,2,2}\ea 
; 16 \kappa^2 \right]  
~.
\end{eqnarray}

\label{th-glasser-onsager}
\end{theorem}

Below we give rigorous proofs of Theorems~\ref{th-main} and
\ref{th-glasser-onsager}, but first we briefly discuss the intuition
and method behind the discovery.  The
logarithm inside the integral in Eq. (\ref{eq-f-ons}) can be
expanded in a Taylor series.  One can perform the definite
integral term by term and then sum them up. Onsager himself
performed this calculation and arrived at the following
expression for the partition function:

\be 
\ln \lambda = \ln (2 \cosh 2K)
 - 
\sum_{n=1}^\infty
 {2n \choose n}^2 {\kappa ^{2n}  \over 4 n} 
~.
\label{eq-ons-series}
\ee
He did not proceed further, except to note the finite radius of
convergence.  

We instead approached the sum in (\ref{eq-ons-series}) as a {\it
  formal power series}.  Unlike normal power series, formal power
series are defined algebraically, independently of convergence.
Rather than interpreting this series analytically as converging to a
function, we instead attempted to match the coefficients in the series
with those in the formal power series definitions of special
functions.  If all coefficients match, then the two series are equal
in an algebraic sense and we will have found the desired special
function.  Theorems \ref{th-main} and \ref{th-glasser-onsager} in fact
follow from (\ref{eq-ons-series}).

It is interesting to note that although  Eqs (\ref{eq-f-ons}) and
(\ref{eq-ons-series}) above appear together as Eq. (109c) in Onsager's
1944 paper, yet (\ref{eq-f-ons}) is widely known whereas
(\ref{eq-ons-series}) has received very little attention in
comparison.
When we recently came across (\ref{eq-ons-series}) for the first time,
we immediately suspected that the series was a $_pF_q$  generalized
hypergeometric function.  Even with only passing familiarity with
special functions, readers will recognize the following clues pointing
to a generalized hypergeometric function: (i) a power series, (ii)
factorials in the numerator and denominator of the series coefficients
and (iii) the arguments of the factorials grow with the degree of the
monomials.

Onsager himself seems to have been
at least partially aware of the connection with hypergeometric
functions, for he wrote in the appendix,

\begin{quote}
We shall deal here with the evaluation of various integrals which
occur in the text. Most of these can be reduced in straightforward
fashion to complete elliptic integrals; {\it only the partition
  function itself is of a type one step higher than the theta
  functions,} [emphasis added] and involves a little analysis which is
not found in textbooks.
\end{quote}
As mentioned earlier, 
the elliptic integrals are special cases of the 
 $_2F_1$ 
ordinary or Gaussian
hypergeometric functions.  
Moreover, the theta and elliptic functions are related to elliptic
integrals or their inverses. Onsager states that the partition
function is ``one step higher'' than the theta functions, 
and indeed
the
$_4F_3$ function that we
find in the evaluation of the partition function is a step or
two more complicated than the $_2F_1$ functions.

The power series method we originally used to arrive at
(\ref{eq-claim2}) from (\ref{eq-ons-series}) is purely algebraic.
Moreover, it is possible to guarantee the correct behavior on and
outside the radius of convergence.  One way around the convergence
issue is analytic continuation.  However, this is in fact not needed
because the series converges even at the critical temperature and can
be evaluated in terms of Catalan's constant.  We will not further
discuss these technical points and will instead take a much
easier-to-understand approach.

We give below an elementary proof.
%
Noting that differentiation is much simpler than integration, we will
differentiate $-\ln \lambda_*$ and then use the fundamental theorem of
calculus, obtaining $u$. So $\ln \lambda$ and $\ln \lambda_*$ differ
only by a constant, which we will show is zero.
We next state and prove 
a hypergeometric identity, from which 
Theorem \ref{th-main} follows easily:

\begin{lemma}
\label{lemma1}
\begin{eqnarray}
&
{_4F_3}
\left[ \ba{c} 2,2,{5\over 2},{5\over 2} \\ {3,3,3}\ea ;x \right]
\nonumber
\\
& \quad = 
- {128\over 9 x^2} 
+{256 {\mathsf K}(\sqrt x)\over 9 \pi x^2}
-{32\over 9 x} ~
{_4}F_3
\left[ \ba{c}1,1,{3\over 2},{3\over 2} \\ {2,2,2}\ea ;x \right]
~.
\label{eq-lemma-rgerg}
\end{eqnarray}

\end{lemma}

\begin{proof}

Using the series definition (\ref{eq-pfq-def}),
the term of degree $n\geq -1$ of 
$\big({_4F_3}
\big[ \scriptsize  \ba{c} 2,2,{5\over 2},{5\over 2} \\ {3,3,3}\ea ;x \big]
+
\frac{32}{ 9 x} ~
{_4}F_3
\big[ \scriptsize \ba{c}1,1,{3\over 2},{3\over 2} \\ {2,2,2}\ea ;x \big]
\big)$
is
\begin{eqnarray}
&
\frac{(2)_n^2(\frac{5}{2})_n^2}{n!(3)_n^3} x^n
+\frac{32}{ 9 x} 
\frac{(1)_{n+1}^2 (\frac{3}{2})_{n+1}^2}{(n+1)!(2)_{n+1}^3} x^{n+1}
\nonumber 
\\
&=
{128 \over 9 \pi }
 \left( {\Gamma(\frac{5}{2} +n) \over 
\Gamma( 3 +n) }\right)^2 x^n 
~.
\end{eqnarray}
Similarly, for the expression  $\frac{256}{ 9 \pi x^2} {\mathsf
  K}(\sqrt x)$, which by (\ref{eq-ellipticK-2f1}) is $\left(\frac{256}{ 18
  x^2}\right)  {_2}F_1 \big[ \scriptsize \ba{c}{\tiny 1/ 2},{\tiny 1/ 2}
  \\ {1}\ea ;x \big] $, we get for $n\geq -2$:

\begin{eqnarray}
&
\frac{256}{ 18
  x^2} 
\frac{(1/2)_{n+2}^2}{(n+2)!(1)_{n+2}} x^{n+2}
=
{128 \over 9 \pi }
 \left( {\Gamma(\frac{5}{2} +n) \over 
\Gamma( 3 +n) }\right)^2 x^n
~. 
\end{eqnarray}
The two series thus agree for the coefficients of $x^{n}$ for every
$n\geq -1$, but for $n=-2$ the $_4F_3$ terms do not contribute.  The
lone $x^{-2}$ term, for $n=-2$ in the above expression, is 
\be
{128 \over 9 \pi x^2 } \left( {\Gamma(\frac{5}{2}-2 ) \over \Gamma( 3-2 )
}\right)^2 ={128\over 9 x^2} ~.   
\ee
This term cancels the term $ - \frac{128}{9x^2}$ in
(\ref{eq-lemma-rgerg}) and the claim follows.
\end{proof}

\begin{proof}[Proof of Theorem  \ref{th-main}]

Observe from (\ref{eq-ons-u-def}) that $-\ln \lambda$ is an
antiderivative of $u$.  Recall that the infinitely many
antiderivatives $\int u(\beta)d\beta$ of the internal energy
$u(\beta)$ differ only by a constant.  Our strategy will be to
show that $-\ln \lambda_*$ is also an antiderivative of $u$ 
and with the same integration constant.

The general formula for the derivative of a ${_4}F_3$ function
is 

\begin{eqnarray}
& \deriv{x} ~
{_4}F_3
\left[ \ba{c}{a_1,a_2,a_3,a_4} \\  {b_1,b_2,b_3}\ea ;x \right]
\nonumber \\
& \quad =  
{a_1 a_2 a_3 a_4 \over b_1 b_2 b_3}
{_4}F_3
\left[ \ba{c}{a_1+1,a_2+1,a_3+1, a_4+1} \\ {b_1+1,b_2+1,b_3+1}\ea ;x \right]
~. 
\end{eqnarray}
This general differentiation rule gives us 
\begin{eqnarray}
& \deriv{x} ~
{_4}F_3
\left[ \ba{c}1,1,{3\over 2},{3\over 2} \\ {2,2,2}\ea ;x \right]
=
{9 \over 32}
~{_4}F_3
\left[ \ba{c} 2,2,{5\over 2},{5\over 2} \\ {3,3,3}\ea ;x \right]
~.
\end{eqnarray}
From the supporting lemma, 
we then  get
\begin{eqnarray}
& \deriv{x} 
{_4}F_3
\left[ \ba{c}{1,1,{3\over 2},{3\over 2}} \\ {2,2,2}\ea ;x \right]
\nonumber
\\
&   \quad =  -{4 \over x^2} + {8 {\mathsf 
 K}(\sqrt{x}) \over \pi x^2 } 
-
{{_4}F_3
\left[ \ba{c}{1,1,{3\over 2},{3\over 2}} \\ {2,2,2}\ea ;x \right]\over x}
\label{eq-3f4-deriv}
~. 
\end{eqnarray}
This result together with the product rule gives us
\begin{eqnarray}
& \deriv{x} ~ \left(
x~{_4}F_3
\left[ \ba{c}1,1,{3\over 2},{3\over 2} \\ {2,2,2}\ea ;x \right]
\right)
=
- {4 \over x}
+{8 {\mathsf K}(\sqrt x)\over \pi x}  
~.
\label{eq-special-deriv}
\end{eqnarray}
Notice that we get cancellation of the terms with the $_4F_3$
functions, leaving only the elliptic integral and an elementary term.
%
%
%
%
{We point out that this result can also be understood in terms of
    an order-four linear differential operator that factors in terms
    of lower order operators, see ref.\cite{newreferee-ref}.}

We will differentiate $\ln \lambda_*$ starting from (\ref{eq-claim2})
using the chain rule. Note, however, that $\kappa$ as a function of $\beta$
is not invertible, because $\kappa(\beta)$ is not monotonic in $\beta$ and so
the inverse function $\beta(\kappa)$ is multivalued.  To get around this
problem,
let us introduce the change of variable
$\arctanh v = 2K $, so that 

\begin{eqnarray}
\nonumber
\tanh 2K = (\coth 2K)^{-1} &=  v\\
\nonumber
\cosh 2K = (\sech 2K)^{-1} &=   {1\over \sqrt{1-v^2}} \\
\nonumber
\sinh 2K = (\csch 2K)^{-1} &=   {v\over \sqrt{1-v^2}} \\
\kappa=\frac{1}{2} \tanh 2K \sech 2K &= {1\over 2} v \sqrt{1-v^2}
\label{eq-translate231tr3prjnnnnnnnnn}
~. 
\end{eqnarray}
%
%
%


{The hyperbolic function $\cosh 2K$ and the quantity
  $\kappa$ in (\ref{eq-claim2}) can be eliminated by re-expressing them in
  terms of the new variable $v$. Hence, direct substitution of
  (\ref{eq-translate231tr3prjnnnnnnnnn})
  into (\ref{eq-claim2}) gives us,}

\begin{eqnarray}
\ln \lambda_*&= 
\ln {2\over \sqrt{1-v^2}}
\nonumber \\
&\quad 
 -
{v^2 (1 - v^2 ) \over 4}
{_4}F_3
\left[ \ba{c}{1,1,{3\over 2},{3\over 2}} \\ {2,2,2}\ea ;   4 v^2 (1 - v^2) \right]
~.
\end{eqnarray}
We now use (\ref{eq-special-deriv}) to calculate $d \ln \lambda_*/dv$ using the
chain rule for derivatives. We get, after simplification,  
\be
{\partial \ln \lambda_*\over \partial v}
=
{\pi -\left(2-4 v^2\right) {\mathsf  K} (2v\sqrt{1-v^2}) \over 2 \pi  v \left(1-v^2\right)}
~.
\label{eq-woreijngeoirijn}
\ee
Note how the differentiation of $\ln \lambda_*$ thus leads to a remarkable
reduction of type $_4F_3$ to type $_2F_1$.
We next use the chain rule again:
\bel{eq-also0e0bt9h}
 {\partial \ln \lambda_* \over
  \partial \beta } =
{\partial \ln \lambda_* \over \partial v~~} 
{\partial v \over \partial   K} 
{\partial K \over \partial   \beta} 
= 2 (1-v^2) J~ {\partial \ln \lambda_* \over \partial v~~} 
~. \ee
From (\ref{eq-woreijngeoirijn}) and  (\ref{eq-also0e0bt9h}) we get the
following final expression:
\be
-{\partial \ln \lambda_*\over \partial \beta}
= 
-{J \over v} \left(1+ {2 \over \pi   } \left(2 v^2 -1\right) {\mathsf  K} (2v\sqrt{1-v^2})
\right)
~.
\label{eq-wow!!!}
\ee
Observe that (\ref{eq-wow!!!})  and (\ref{eq-ons-u}) are identical,
after changing variables using 
(\ref{eq-translate231tr3prjnnnnnnnnn}).  We have thus shown that $-\ln
\lambda_*(\beta)$ is an antiderivative of $u(\beta)$.  We are almost
done.

Since $-\ln \lambda_*$ and $-\ln \lambda$ are both
antiderivatives of $u$ with respect to $\beta$, therefore by the
fundamental theorem of calculus they differ only by a real
constant. Let \mbox{$C=\ln \lambda- \ln \lambda_*$}.  To show equality of
$\lambda_*$ and $\lambda$, it suffices to show that $C=0$.  We can
calculate $C$ explicitly from the values of both $\lambda_*$ and
$\lambda$ for some convenient value of $\beta$.  The easiest choice is
$\beta=0$, for which we get
$
\ln \lambda_*(0) = \ln \lambda(0) = \ln 2~.
$
So $C=0$ and the
claim follows. \qed

\end{proof}

\begin{proof}[Proof of Theorem \ref{th-glasser-onsager}]

We will proceed as with Lemma \ref{lemma1} and derive a
hypergeometric identity, from which Theorem \ref{th-glasser-onsager}
will follow immediately.

On the one hand, the coefficient of the term of degree $\kappa^{2n}$ for all $n\geq 1$ of the
quantity
\[
 \kappa^2 ~{_4}F_3
\left[ \ba{c}{1,1,{3\over 2},{3\over 2}} \\ {2,2,2}\ea 
; 16 \kappa^2 \right]  
\]
is
\begin{eqnarray}
&
16^{n-1}
\frac{(1)_{n-1}^2 (\frac{3}{2})_{n-1}^2}{(n-1)!(2)_{n-1}^3} 
&=
{4^{2n-1}  \over \pi n^3  } \left(\Gamma(n+ \frac 1 2  ) \over \Gamma(n)\right)^2
~.\label{eq-1-98y4tb7y}
\end{eqnarray}
On the other hand, 
the 
quantity
\be
 {1\over\pi} {\mathsf E(4 \kappa)} 
+ \kappa^2 ~{_4}F_3
\left[ \ba{c}{{1\over 2},1,1,{3\over 2}} \\ {2,2,2}\ea 
; 16 \kappa^2 \right]  
\ee
can be rewritten completely in terms of hypergeometric functions as
\be
 \left({1 \over 2}\right) ~
{_2F_1}\left[ \ba{c}{{1\over 2},-{1\over 2}} \\ {1}\ea ; 16\kappa^2 \right]
+ \kappa^2 ~{_4}F_3
\left[ \ba{c}{{1\over 2},1,1,{3\over 2}} \\ {2,2,2}\ea 
; 16 \kappa^2 \right]  
~.
\ee
The term of degree zero in the series above is $1/2$ because only the $_2F_1$
term contributes.  For all $n\geq 1$
the term of degree $2n$ has coefficient

\bea
&{1\over 2} {(\frac 1 2 )_n(- \frac 1 2 )_n \over n! (1)_n} 
+ 16^{n-1} {(\frac 1 2)_{n-1}  (1)_{n-1}^2 (\frac 3 2)_{n-1}   \over (n-1)! (2)_{n-1}^3}
\nonumber\\
&=-{4^{2n-1} \over \pi} {\Gamma(n-\frac 1 2 )\Gamma(n+ \frac 1 2)\over [\Gamma(n+1)]^2 }
+ {2^{4n-3} \over \pi n^3}   {\Gamma(n-\frac 1 2 )\Gamma(n+\frac 1 2 )\over [\Gamma(n)]^2 }
\nonumber\\
&= 
- {4^{2n-1}\over \pi n^3} \left({\Gamma(n+\frac 1 2 )  \over \Gamma(n)}\right)^2
~.\label{eq-2-98y4tb7y}
\eea

Observing that (\ref{eq-1-98y4tb7y}) and (\ref{eq-2-98y4tb7y}) are
identical except for sign for all $n\geq 1$ and then taking into
account separately the previously mentioned $n=0$ term, we arrive at 

\be
 \kappa^2 ~{_4}F_3
\left[ \ba{c}{1,1,{3\over 2},{3\over 2}} \\ {2,2,2}\ea 
; 16 \kappa^2 \right]  
= \frac 1 2 
- {1\over\pi} {\mathsf E(4 \kappa)} 
- \kappa^2 ~{_4}F_3
\left[ \ba{c}{{1\over 2},1,1,{3\over 2}} \\ {2,2,2}\ea 
; 16 \kappa^2 \right]  
~.~~
\ee
The claim follows from substitution of this hypergeometric identity
into Theorem \ref{th-main}.~\qed 

\end{proof}

\ack We thank T.~M.~Viswanathan for discussions and for suggesting the
study of formal power series.  We thank J. C. Cressoni and S. Salinas
for feedback and CNPq for funding. We are grateful to the anonymous
referees for helpful suggestions.

\end{document}